\documentclass[aip, amsmath, amssymb, reprint]{revtex4-1}

\usepackage{graphicx}
\usepackage{dcolumn}
\usepackage{bm}

\usepackage[utf8]{inputenc}
\usepackage[T1]{fontenc}
\usepackage{mathptmx}
\usepackage{etoolbox}
\usepackage{mathrsfs}
\usepackage{soul, xcolor}
\usepackage{color, wasysym}

\usepackage{amsmath}
\usepackage{amsfonts}
\usepackage{amssymb}
\usepackage{graphicx}
\usepackage{bm}
\usepackage{textcomp}
\usepackage{srcltx}
\usepackage{amsthm} 
\usepackage{latexsym}
\usepackage{epsfig}
\usepackage{mathrsfs}

\makeatletter
\def\@email#1#2{%
 \endgroup
 \patchcmd{\titleblock@produce}
  {\frontmatter@RRAPformat}
  {\frontmatter@RRAPformat{\produce@RRAP{*#1\href{mailto:#2}{#2}}}\frontmatter@RRAPformat}
  {}{}
}%

\DeclareMathOperator{\D}{d\!}
\DeclareMathOperator{\E}{e} 
\DeclareMathOperator{\I}{i}

\newtheorem{theorem}{Theorem}
\newtheorem{definition}{Definition}

\makeatother
\begin{document}


\title[Probability density functions as solutions of heterogeneous Cattaneo-Vernotte diffusion equation]{Probability density functions as solutions of heterogeneous Cattaneo-Vernotte diffusion equation}

\author{K. G\'{o}rska}
\email{katarzyna.gorska@ifj.edu.pl}
\affiliation{H. Niewodnicza\'{n}ski Institute of Nuclear Physics, Polish Academy of Sciences, Division of Theoretical Physics, ul. Eliasza-Radzikowskiego 152, PL 31-342 Krak\'{o}w, Poland}

\author{\v{Z}. Tomovski}
\affiliation{Department of Mathematical Analysis and Applications of Mathematics, Faculty of Science, Palacký University Olomouc, 17 listopadu 1192/12, 779 00 Olomouc, Czech Republic}

\author{A. Horzela}
\affiliation{H. Niewodnicza\'{n}ski Institute of Nuclear Physics, Polish Academy of Sciences, Division of Theoretical Physics, ul. Eliasza-Radzikowskiego 152, PL 31-342 Krak\'{o}w, Poland}

\author{M. Magdziarz}
\affiliation{Hugo Steinhaus Center, Faculty of Pure and Applied Mathematics, Wrocław University of Science and Technology, Wyb. Wyspiańskiego 27, 50-370 Wrocław, Poland}


\begin{abstract}
In this paper, we considered a heterogeneous Cattaneo-Vernotte equation with an exponential type of diffusion coefficient under the fundamental initial and boundary conditions stating that the solution vanishes at $+/-$ infinity. Owing to the Laplace transform method we obtain two forms of exact analytical solutions which are presented in terms of the ratio of modified Bessel functions.  Using the theory of complete monotone functions, we show that the obtained solutions are probability density functions.
\end{abstract}

\begin{keywords}
{completely monotonicity, ratio of modified Bessel functions, diffusion phenomena}
\end{keywords}

\maketitle

\section{Motivation}

In diffusion processes in heterogeneous media, characterized by a spatially varying diffusion coefficient, the class of completely monotonic functions (CMs) plays a special role. CM functions are often used to demonstrate that the solutions of evolution equations arising from the diffusion approach, such as the Fokker-Planck (FP) or Cattaneo-Vernotte (CV) equations, are non-negative. The non-negativity of a function whose Laplace transform exists can be shown, e.g., by using Bernstein's theorem, which relates a CM function to a non-negative function via the Laplace transform. For example, from purely mathematical side, it was proved in \cite{EGrosswald76, MEHIsmail76, MEHIsmail77, MHEIsmail90} that
\begin{equation}
    \label{22/05/26-6}
    \frac{1}{\sqrt{s}} \frac{K_{\nu-1}(\sqrt{s})}{K_{\nu}(\sqrt{s})}, \quad s > 0, \quad \nu \geq 0, 
\end{equation}
has a CM character, is infinitely divisible, and corresponds to the student $t$-distribution. From a physical point of view, in \cite[Chapter 5]{TSandev22b} the authors showed that \eqref{22/05/26-6} can be related to the (1+1)-dimensional FP equation in heterogeneous media with ${\cal D}(x) \propto |x|^\mu$, $x\in\mathbb{R}$ and $\mu \leq 2$.

Motivated by \cite{EGrosswald76, MEHIsmail76, MEHIsmail77, MHEIsmail90}, we shall generalize \eqref{22/05/26-6} by introducing $a \equiv a(x)$ and $b \equiv b(x)$. Moreover, we shall modify the arguments of the ratio of the modified Bessel functions, such that it will contain $(\tau s^2 + s)^{1/2}$ instead of $s^{1/2}$, and add to it the prefactor $(\tau + 1/s)^{1/2} = (\tau s + 1)(\tau s^2 + s)^{-1/2}$ instead of $s^{-1/2}$. Namely, we will study
\begin{multline}
    \label{22/05/26-7}
    \frac{\tau s + 1}{\sqrt{\tau s^2 + s}} \frac{I_{\alpha - 1}(a \sqrt{\tau s^2 + s})}{I_\alpha(b \sqrt{\tau s^2 + s})} \qquad \text{and} \\ \frac{\tau s + 1}{\sqrt{\tau s^2 + s}} \frac{K_{\alpha - 1}(b \sqrt{\tau s^2 + s})}{K_\alpha(a \sqrt{\tau s^2 + s})},
\end{multline}
where $0 < a \leq b$, $s > 0$, and $\alpha\in[0,1]$. Since $s$ is related to the first time derivative, then $\tau s^2 + s$ should correspond to the second and first time derivatives. Moreover, following \cite{KGorska26} we will show that \eqref{22/05/26-7} is proportional to the probability density fucntion (PDF; non-negative and normalized) and solve
\begin{multline}
    \label{22/05/26-8}
    \tau \partial^2_t p(x, t) + \partial_t p(x, t) \\ = \partial_x\left\{{\cal D}_\lambda(x)^{1-\alpha} \partial_x\Big[{\cal D}_\lambda(x)^{\alpha} p(x, t)\Big]\right\}  
\end{multline}
for $x\in\mathbb{R}$, $t>0$ and with $p(x, t=0) = \delta(x)$, $\dot{p}(x, t)|_{t=0} = 0$ and $\lim_{x\to\pm\infty} p(x, t) = 0$. In \eqref{22/05/26-8} $\tau > 0$ denotes the time lag, $\alpha\in[0, 1]$ is the stochastic interpretation parameter \cite{JMSancho11, APacheoPozo24}, and ${\cal D}_\lambda(x) = B \exp(\lambda |x|)$, in which $\lambda \in \mathbb{R}$ and $B > 0$. 

\section{Solutions of \eqref{22/05/26-8} with ${\cal D}_\lambda(x) = B \exp(\lambda |x|)$}\label{sect2}

The solutions of \eqref{22/05/26-8} with ${\cal D}_\lambda(x)$, can be found in three steps. In the first one, we take the Laplace transform of \eqref{22/05/26-8}, which enables us to get
\begin{multline}
    \label{13/02/26-2}
    \tau\big(s^2 \widehat{p}_\lambda(\alpha; x, t) - s \delta(x)\big) + s \widehat{p}_\lambda(\alpha; x, t) - \delta(x) = \\ B \E^{\lambda |x|} \partial_x^2 \widehat{p}_\lambda(\alpha; x, t) 
    + (1 + \alpha) \lambda B \E^{\lambda |x|} [2\Theta(x) - 1] \partial_x \widehat{p}_\lambda(x, t) \\ + \alpha \lambda B \E^{\lambda |x|} [\lambda + 2 \delta(x)] \widehat{p}_\lambda(\alpha; x, t).
\end{multline}
Since ${\cal D}_\lambda(x)$ is symmetric about the vertical line through $x=0$, then $p_\lambda(\alpha; x, t)$ should conserve the same symmetry. Therefore, we set $u = \exp(-|x|)$ and $\widehat{p}_\lambda(\alpha; u, s) = \widehat{G}_\lambda(\alpha; s) \widehat{F}_\lambda(\alpha; u, s)$. Equation \eqref{13/02/26-2} in $u$ coordinates can be spitted into two parts: for $u \neq 1$ it reads 
\begin{multline}
    \label{14/02/26-3}
    \partial_u^2 \widehat{F}_\lambda(\alpha; u, s) + [1 - \lambda(\alpha + 1)] \frac{1}{u} \partial_u \widehat{F}_\lambda(\alpha; u, s) + \left(\frac{\alpha \lambda^2}{u^2}\right.  \\ \left. - \frac{\tau s^2 + s}{B} u^{\lambda - 2}\right) \widehat{F}_\lambda(\alpha; u, s) = 0
\end{multline}
and for $u =1$ it is
\begin{multline}
    \label{14/02/26-4}
    -\frac{\tau s + 1}{\widehat{G}_\lambda(\alpha; s)} = \left[-2 B u^{1-\lambda} \partial_u \widehat{F}_\lambda(\alpha; u, s) \right. \\ \left.+ 2 \alpha \lambda B u^{-\lambda} \widehat{F}_\lambda(\alpha; u, s) \right]_{u=1}.
\end{multline}
Equation \eqref{14/02/26-3} is the Bessel differential-type equation whose solution is equal to
\begin{equation}
    \label{14/02/26-5}
    \widehat{F}_{\lambda}(\alpha; u, s) = u^{\frac{\lambda}{2} (1 + \alpha)} Z_{\alpha-1}\left(\frac{2 \I u^{\lambda/2}}{\lambda \sqrt{B}} \sqrt{\tau s^2 + s}\right), 
\end{equation}
where $Z_{\nu}(\sigma)$ is a linear combination of the Bessel functions  \cite{FBrowman58}. Substituting \eqref{14/02/26-5} into \eqref{14/02/26-4} we find that $\widehat{G}_\lambda(\alpha; s)$ reads
\begin{equation*}
    \widehat{G}_\lambda(\alpha; s) = \frac{1}{2\I\sqrt{B}} \frac{\tau s + 1}{\sqrt{\tau s^2 + s}} \left[Z_{\alpha}\left(\frac{2\I}{\lambda\sqrt{B}} \sqrt{\tau s^2 + s}\right)\right]^{-1}.
\end{equation*}

In the final step, we return to the variable $x$ and impose boundary conditions on the solution, i.e., we require that $p_{\lambda}(\alpha; x, t)$ vanishes as $x \to \pm \infty$. That gives two kinds of solutions with respect to the sign of $\lambda$: for $\lambda > 0$ it is
\begin{multline}
    \label{15/02/26-2}
    \widehat{p}_{\lambda; 1}(\alpha; x, s) = \frac{1}{2 \sqrt{B}} \E^{-\frac{\lambda}{2}(1 + \alpha) |x|} \\ \times \frac{\tau s + 1}{\sqrt{\tau s^2 + s}} 
    \frac{I_{\alpha-1}\left(\frac{2}{\lambda \sqrt{B}} \E^{-\lambda |x|/2} \sqrt{\tau s^2 + s}\right)}{I_{\alpha}\left(\frac{2}{\lambda \sqrt{B}} \sqrt{\tau s^2 + s}\right)}
\end{multline}
and for $\lambda < 0$ we have
\begin{multline}
    \label{15/02/26-3}
    \widehat{p}_{\lambda; 2}(\alpha; x, s) = \frac{1}{2 \sqrt{B}} \E^{\frac{|\lambda|}{2}(1 + \alpha) |x|} \\ \times \frac{\tau s + 1}{\sqrt{\tau s^2 + s}}  \frac{K_{\alpha-1}\left(\frac{2}{|\lambda| \sqrt{B}} \E^{|\lambda x|/2} \sqrt{\tau s^2 + s}\right)}{K_{\alpha}\left(\frac{2}{|\lambda| \sqrt{B}} \sqrt{\tau s^2 + s}\right)}.
\end{multline}
In Sect. \ref{sect3} we will show that their inverse Laplace transforms are PDFs for different range of $\alpha$, namely $p_{\lambda; 1}(\alpha; x, t)$ is PDF for $\alpha\in[1/2, 1)$ and $p_{\lambda; 2}(\alpha; x, t)$ is PDF for $\alpha\in[0, 1/2]$.

\section{$p_{\lambda}(\alpha; x, t)$ as a probability density function}\label{sect3}

In the paper, we will use the following definitions and properties taken, e.g., from \cite{RLSchilling10}.
\begin{definition}\label{def1}
$f:(0, \infty) \mapsto \mathbb{R}$ is a CM function if $f\in {\cal C}^\infty$ is non-negative and $(-1)^n f^{(n)}(s) > 0$ for all $s > 0$ and $n\in\mathbb{N}_0$. 
\end{definition}
\begin{definition}\label{def2}
    $g:(0, \infty) \mapsto \mathbb{R}$ is a Bernstein (B) function if $g\in{\cal C}^\infty$ is non-negative and its derivative is CM. 
\end{definition}

\noindent
Property 1: The product of two CM functions is a CM function.

\noindent 
Property 2: The composition of a CM function with a B function is a CM function.

\noindent
Property 3: The linear combination of CM functions on a convex cone is a CM function.

\subsection{Non-negativity of $p_{\lambda; 1}(\alpha; x, t)$ for $\alpha\in[1/2, 1)$}\label{sect3.1}

Due to Bernstein's theorem, $p_{\lambda; 1}(\alpha; x, t)$ is non-negative when its Laplace transform given by \eqref{15/02/26-2} is a CM function. Thus, to prove the CM character of $\widehat{p}_{\lambda; 1}(\alpha; x, s)$ for $\alpha\in[1/2, 1)$ we use Property 1 and present it as the product of two functions, namely 
\begin{multline}\label{1/04/26-1}
\widehat{f}_1(s) = \frac{\tau s + 1}{\sqrt{\tau s^2 + s}} = \sqrt{\tau} \left(1 + \frac{1}{\tau s}\right)^{1/2} \quad \text{and} \\ \widehat{f}_{2,1}(\alpha; x, s) = \frac{I_{\alpha - 1}\Big(\frac{2}{\lambda \sqrt{B}} \E^{- \lambda|x|/2}\,  \sqrt{\tau s^2 + s}\Big)}{I_{\alpha}\Big(\frac{2}{\lambda \sqrt{B}}\,  \sqrt{\tau s^2 + s}\Big)}.
\end{multline}
As is shown in \cite[Eq. (1.19)]{KSMiller01} and \cite[Sect.4 A]{KGorska20}, $\widehat{f}_1(s)$ is a CM function and $(\tau s^2 + s)^{1/2}$ is a B function, respectively. Therefore, by Property 2, $\widehat{f}_{2, 1}(\alpha; x, s) = \widehat{F}_{2, 1}(\alpha; x, \sqrt{\tau s^2 + s})$ will be a CM function if 
   $ \widehat{F}_{2,1}(\alpha; x, s) = \big[I_{\alpha-1}(a s)/I_{\alpha}(a s)\big]\; \big[I_{\alpha}(a s)/I_{\alpha}(b s)\big]$, is a CM function. Here, $a = b \E^{-\lambda |x|/2}$ and $b = 2/(\lambda \sqrt{B})$ such that $0 < a \leq b$.

\begin{theorem}
 $I_{\alpha-1}(s)/{I_{\alpha}(s)}$ is a CM function for all $\alpha\in[1/2, 1)$ and $s >0$.\end{theorem}
\noindent 
{\em Proof.} To show that $I_{\alpha-1}(s)/{I_{\alpha}(s)}$ is a CM function for $\alpha\in[1/2, 1)$ and $s >0$ let us observe that it is a positive function, see \cite[Theorem 1]{JSegura21}. Next, applying the Mittag-Leffler expansion for a meromorphic function, where $j_{\alpha,1} < j_{\alpha,2}<...<j_{\alpha,n}<...$ are the positive zeros of the Bessel function $J_{\alpha}(x)$:  $I_{\alpha+1}(s)/I_\alpha(s) =  \sum_{n=1}^\infty 2 s /(s^2 + j^2_{\alpha, n})$ and using the recurrence relation, we get
\begin{equation}
    \label{20/04/26-1}
    \frac{I_{\alpha-1}(s)}{I_\alpha(s)} = \frac{2\alpha}{s} + \frac{I_{\alpha + 1}(s)}{I_\alpha(s)} 
    = \frac{2\alpha}{s} + \sum_{n=1}^\infty \frac{2 s}{s^2 + j^2_{\alpha, n}}. 
\end{equation}
The trick $s/(s^2+j^2_{\alpha,n}) = \int_0^{\infty} \exp\big[-\big(s + j^2_{\alpha,n}/s\big)u\big] \D u$, enables us to rewrite \eqref{20/04/26-1} in the form
\begin{align*}
     \frac{I_{\alpha-1}(s)}{I_\alpha(s)} & = \frac{2\alpha}{s} + \frac{I_{\alpha + 1}(s)}{I_\alpha(s)} = \frac{2\alpha}{s} + 2\int_0^{\infty} \E^{-su} \left\{\sum_{n=1}^{\infty} \E^{-\frac{j^2_{\alpha,n}}{s}u} \right\} \D u.
\end{align*}
The series in curly brackets is the Jacobi theta function, which converges for all $s > 0$ and is non-negative. Therefore, by Bernstein's theorem and Property 3, the proof is complete. \qed

\noindent
{\bf Remark 1.}
  {\rm  For $\alpha = 1/2$ we can express $I_{\alpha-1}(s)/I_{\alpha}(s)$ as
    \begin{multline*}
        \frac{I_{\alpha-1}(s)}{I_\alpha(s)} = \frac{2\alpha}{s} + \sum_{n=1}^\infty\left\{\frac{1}{s+ \I j_{\alpha,n}}+\frac{1}{s- \I j_{\alpha,n}}\right\} 
    \\ = \frac{2\alpha}{s} + \sum_{n=1}^\infty \left\{\int_0^{\infty} \E^{-\xi s}\big(\E^{\I \xi j_{\alpha,n}} + \E^{-\I \xi j_{\alpha,n}}\big)\right\} \D \xi \\
     = 2\int_0^{\infty} \E^{-\xi s}\left\{\alpha + \sum_{n=1}^{\infty} \cos(\xi j_{\alpha,n})\right\} \D \xi. 
    \end{multline*}
The numerical calculations show that the series $\alpha + \sum_{n=1}^{\infty}\cos(\xi j_{\alpha,n}) = \delta_\alpha(\xi)$ is positive only for $\alpha=1/2$, when it reduces to the Dirac comb $\delta_{1/2}(\xi)=\sum_{n\in\mathbb{Z}}\delta(\xi-2n)$, where $\delta(x)$ is the $\delta$ Dirac distribution, see Fig. \ref{Dirac_Comb}.
Hence, 
\begin{equation*}
 \frac{I_{-1/2}(s)}{I_{1/2}(s)} = 2\int_0^{\infty} \E^{-\xi s}\delta_{1/2}(\xi) \D \xi = 2\int_0^{\infty}\E^{-\xi s} \D\mu(\xi),
\end{equation*}
where it is a Lebesgue integral with respect to the non-negative finite Borel measure $\mu(\xi)$ on $(0,\infty)$ with cumulative distribution function $\delta(\xi)$.}

\begin{figure}
\begin{center}
\includegraphics[width=8.5cm]{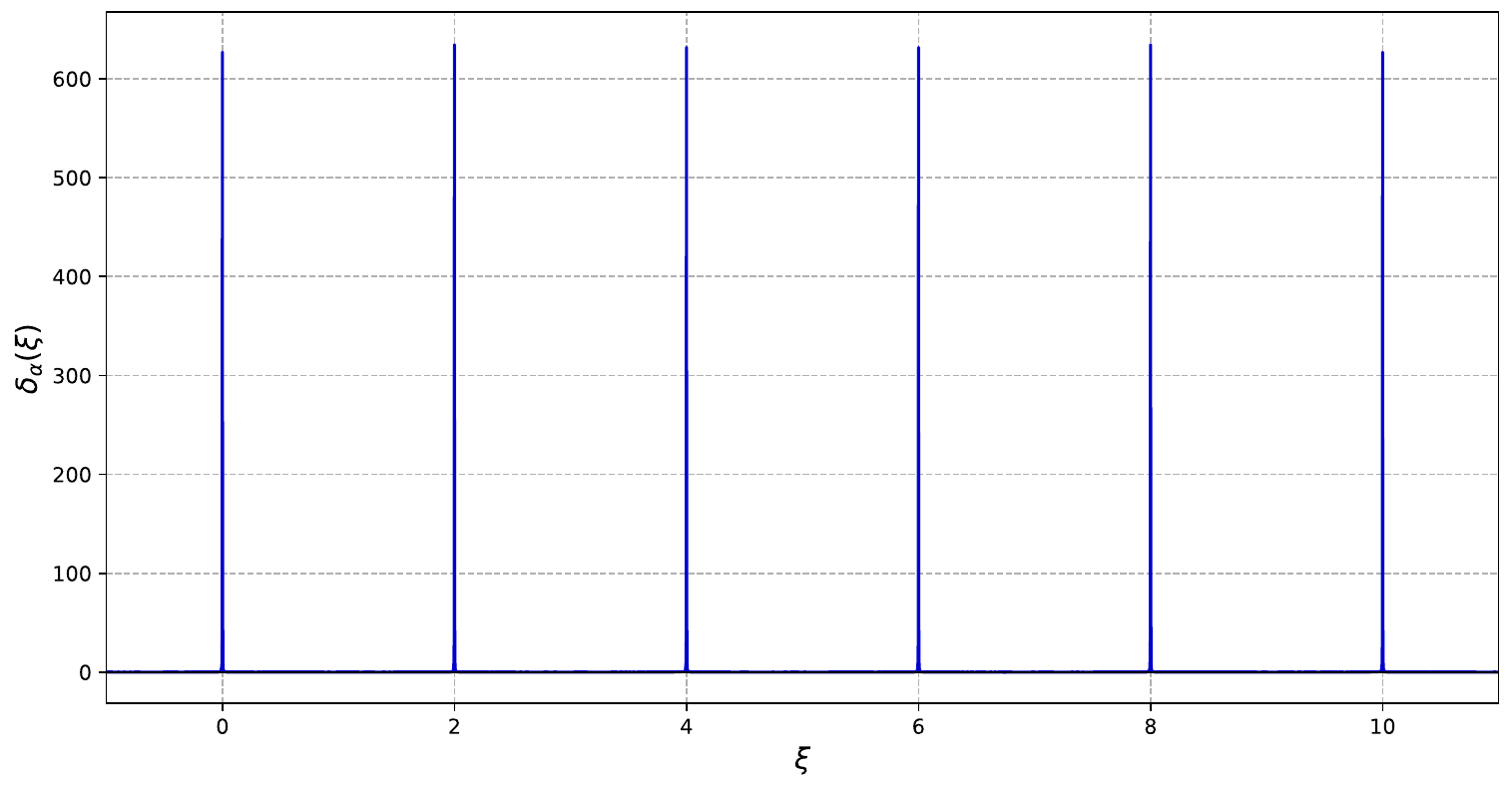}
\caption{Graph of the series $\delta_\alpha(\xi)=\alpha+\sum_{n=1}^{\infty}\cos(\xi j_{\alpha,n})$ for $\alpha=1/2$ using Abel regularization.}
\label{Dirac_Comb}
\end{center}
\end{figure}
\begin{theorem}
    $I_{\alpha}(a s)/I_{\alpha}(b s) = \E^{- \widehat{g}_1(s)}$ is a CM function for $b\geq a>0$ and $\alpha\in[1/2, 1)$ if $\widehat{g}_1(s) = \ln\big[I_{\alpha}(b s)/I_{\alpha}(a s)]$ is a B function.
\end{theorem}
\noindent
{\em Proof.} Firstly, notice that $\widehat{g}_1(s)$ is a non-negative function, which comes from the algebraic inverse of \cite[Eq. (2.1)]{ABaricz10}. Direct calculations show that its first derivative is a CM function. Indeed,
\begin{align*}
\frac{\D}{\D s} &\widehat{g}_1(s) 
= b \frac{I_{\alpha-1}(b s)}{I_\alpha(b s)} - a \frac{I_{\alpha - 1}(a s)}{I_\alpha(a s)} \nonumber \\ & 
= b\left(\frac{2\alpha}{b s} + \sum_{n=1}^{\infty} \frac{2b s}{(bs)^2 + j^2_{\alpha,n}}\right) - a\left(\frac{2\alpha}{as}+\sum_{n=1}^{\infty} \frac{2as}{(as)^2+j^2_{\alpha,n}}\right) \nonumber \\
& = 2\sum_{n=1}^{\infty} \frac{s}{s^2 + \frac{j^2_{\alpha,n}}{b^2}} - 2\sum_{n=1}^{\infty}\frac{s}{s^2 + \frac{j^2_{\alpha,n}}{a^2}} \nonumber \\ &
= 2\left(\frac{1}{a^2} - \frac{1}{b^2}\right)\frac{1}{s} \sum_{n=1}^{\infty} \frac{s}{s^2 + \frac{j^2_{\alpha,n}}{b^2}} \cdot \frac{s}{s^2 + \frac{j^2_{\alpha,n}}{a^2}} j^2_{\alpha,n} > 0.
\end{align*}
Using the integral representation below \eqref{20/04/26-1}, we have
\begin{align}\label{17/05/26-5}
\frac{\D}{\D s} &\widehat{g}_1(s) = \left(\frac{1}{a^2} - \frac{1}{b^2}\right)\frac{2}{s}\iint\limits_0^{\infty} \E^{-s(u_1+u_2)}\delta_{\alpha}(u_1,u_2,s) \D u_1 \D u_2 \nonumber \\ &
 = \left(\frac{1}{a^2} - \frac{1}{b^2}\right)\frac{2}{s} \int_0^\infty \E^{-s v} \left[\int_0^v \delta_\alpha(v-u_2, u_2, s) \D u_2\right] \D v,
\end{align}
where $\delta_{\alpha}(u_1,u_2,s) = \sum_{n=1}^{\infty} j^2_{\alpha,n} \exp\big[-(u_1/a^2 + u_2/b^2) j^2_{\alpha,n}/s\big]$ is a spectral function expressed in the form of a convergent series connected to the Jacobi theta function: $ q \theta'_3(0,q)=2\sum_{n=1}^{\infty}n^2q^{n^2}$, $q\in(0, 1)$. Note that the square bracket in \eqref{17/05/26-5} is non-negative for $b \geq a$, thus by Bernstein's theorem $\widehat{g}_1(s)$ is a B function. The use of Property 2 completes the proof. \qed

\subsection{Non-negativity of $p_{\lambda; 2}(\alpha; x, t)$ for $\alpha\in[0, 1/2]$}\label{sect3.2}

Analogously to Sect. \ref{sect3.1}, the non-negativity of $p_{\lambda; 2}(\alpha; x, t)$ is proved using Bernstein's theorem. Thus, we present $\widehat{p}_{\lambda; 2}(\alpha; x, s)$ as the product $\widehat{f}_1(s)$ given by \eqref{1/04/26-1} and
\begin{align*}
    \widehat{f}_{2,2}(\alpha; x, s) & = \frac{K_{\alpha - 1}\Big(\frac{2}{\lambda \sqrt{B}} \E^{|\lambda x|/2}\, \sqrt{\tau s^2 + s}\Big)}{K_{\alpha}\Big(\frac{2}{\lambda \sqrt{B}} \sqrt{\tau s^2 + s}\Big)} \\ & = \frac{K_{1 - \alpha}\Big(\frac{2}{\lambda \sqrt{B}} \E^{|\lambda x|/2}\, \sqrt{\tau s^2 + s}\Big)}{K_{\alpha}\Big(\frac{2}{\lambda \sqrt{B}} \sqrt{\tau s^2 + s}\Big)},
\end{align*}
where $K_{\alpha}(\cdot) = K_{-\alpha}(\cdot)$ and $K_{\alpha-1}(\cdot) = K_{1-\alpha}(\cdot)$. To show that the function $ \widehat{f}_{2,2}(\alpha; x, s)$ for $\alpha\in[0, 1/2]$ is a CM function, we can use the results presented in \cite[Appendix D]{KGorska26}. However, in the paper, we will take a different approach. First, $\widehat{f}_{2, 2}(\alpha; x, s) = \widehat{F}_{2, 2}(\alpha; x, \sqrt{\tau s^2 + s})$, in which $\widehat{F}_{2, 2}(\alpha; x, s) = \big[K_{1-\alpha}(b s)/K_\alpha(b s)\big]\, \big[K_\alpha(b s)/K_\alpha(a s)\big]$, where $b = a \exp(|\lambda x|/2)$ and $a = 2/(\lambda \sqrt{B})$ such that $0 < a \leq b$. Knowing that $(\tau s^2 + s)^{1/2}$ is a B function we shall show that $\widehat{F}_{2, 2}(\alpha; x, s)$ is a CM function.

\noindent
{\bf Lemma 1.} The function $h(\alpha; \xi)= \big\{\xi\big[J^2_{\alpha}(\xi) + Y^2_{\alpha}(\xi)\big]\big\}^{-1}$ for $\xi>0$ and $\alpha\in[0,1/2)$ is convex.

\medskip
 \noindent
{\em Proof.} 
$h(\alpha; \xi)$ is convex for $\alpha\in[0,1/2)$ if its algebraic inverse $h_1(\alpha; \xi) = \xi\big[J^2_{\alpha}(\xi) + Y^2_{\alpha}(\xi)\big] > 0$ is concave. The $(n-1)$-derivative of $h'_1(\alpha; \xi)$ is given in \cite[Eq. (3.3)]{LLorch63} and reads 
\begin{multline}\label{17/06/26-3}
 (h'_1)^{(n-1)}(\alpha; \xi) = \frac{8}{\pi^2} \int_0^{\infty} K_0^{(n-1)}(2\xi \sinh u) (2 \sinh u)^{n-1} \\ \times (\tanh u) \cosh (2\alpha u) [\tanh u - 2\alpha \tanh(2\alpha u)] \D u,
\end{multline}
where $K_0(\cdot)$ is MacDonald's function of zero order. He \cite{WatB} points out, that the square bracket in \eqref{17/06/26-3} is negative for $|\alpha| > 1/2$, and positive for $|\alpha| < 1/2$. Therefore, $h_1(\alpha; \xi)$ decreases when $|\alpha| >1/2$ and increases when $\alpha\in[0,1/2)$. 

According to \cite[Theorem 4]{KSMiller01} and CM character of $K_0(\cdot)$ (see the paragraph surrounding \cite[Eq. (3.3)]{LLorch63}) it can be concluded that $h'(\alpha; \xi)$ is a CM function then $h_1(\alpha; \xi)$ is a B function for $\alpha \in [0, 1/2)$ which means that $h_1(\alpha; \xi)$ is concave for $\alpha\in[0,1/2)$. \qed
\begin{theorem}
 \label{mm}
$K_{\alpha-1}(s)/K_{\alpha}(s)$  is a CM function for $\alpha\in [0,1/2)$ and $s>0$.
\end{theorem}
\noindent 
{\em Proof.} 
Grosswald in \cite{EGrosswald76} proved that 
\begin{equation}\label{13/04/26-5}
  \frac{K_{\alpha-1}(s)}{K_{\alpha}(s)} = \frac{4}{\pi^2}\int_0^{\infty} \frac{s}{s^2 + \xi^2} \,    \frac{\D \xi}{\xi \big[J^2_{\alpha}(\xi) + Y^2_{\alpha}(\xi)\big]}
\end{equation}
for $\alpha\geq 0$ and $s>0$. Then, using the Laplace transform integral $s/(s^2 + \xi^2) = \int_0^{\infty} \E^{-s u} \cos(\xi u) \D u$ 
and interchanging the order of integrations in \eqref{13/04/26-5}, we get
\begin{equation*}
\frac{K_{\alpha-1}(s)}{K_{\alpha}(s)} = \frac{4}{\pi^2}\int_0^{\infty} \E^{-s u} \left\{\int_0^{\infty}\frac{\cos(\xi u)\D\xi}{\xi\big[J^2_{\alpha}(\xi)+Y^2_{\alpha}(\xi)\big]}\right\}\D u.
\end{equation*}
By Tuck \cite{Tuc06}, the spectral function in the Laplace transform is the cosine Fourier transform, which is positive, if the integrand $\xi\mapsto \xi^{-1} \big[J^2_{\alpha}(\xi) + Y^2_{\alpha}(\xi)\big]^{-1}$ has a positive second derivative, that is, that everywhere-convex functions have everywhere-positive Fourier-cosine transforms. Applying Lemma 1, we complete the proof. \qed
\begin{theorem}
    \label{ff}
    $K_{\alpha}(b s)/K_{\alpha}(a s) = \E^{- \widehat{g}_2(s)}$ is a CM function for $b\geq a > 0$ and $\alpha\in[0, 1/2)$ if $\widehat{g}_2(s) = \ln\big[K_{\alpha}(a s)/K_{\alpha}(b s)]$ is a B function.
\end{theorem}
\noindent
{\em Proof.} Due to \cite[Eq. (3.6)]{ABaricz10} $K_{\alpha}(a s)/K_{\alpha}(b s) > 1$ for $\alpha \in[0, 1/2)$ thus, $g_2(s)$ is non-negative. The first derivative of $g_2(s)$ reads
\begin{align*}
g_2'(s) & = b\frac{K_{\alpha-1}(bs)}{K_{\alpha}(bs)}-a\frac{K_{\alpha-1}(as)}{K_{\alpha}(as)} \\ & = \frac{1}{s}\left(b s \frac{K_{\alpha-1}(bs)}{K_{\alpha}(bs)} - a s \frac{K_{\alpha-1}(as)}{K_{\alpha}(as)}\right).    
\end{align*}
Since $s\mapsto\frac{s^2}{s^2+\xi^2}$ is increasing in s, using the integral representation \eqref{13/04/26-5}, we obtain
\begin{align}
g'_2(s) & = \frac{4}{\pi^2s} \int_0^{\infty}\left\{\frac{(bs)^2}{(bs)^2+\xi^2} - \frac{(as)^2}{(as)^2+\xi^2}\right\} \frac{\D\xi}{\xi\big[J^2_{\alpha}(\xi) + Y^2_{\alpha}(\xi)\big]} \nonumber \\ 
& = \frac{4}{\pi^2} \frac{b^2-a^2}{a^2b^2} \frac{1}{s} \int_0^{\infty} \frac{s}{s^2 + \xi^2/b^2}\, \frac{s}{s^2+\xi^2/a^2}\, \frac{\xi \D\xi}{J^2_{\alpha}(\xi) + Y^2_{\alpha}(\xi)} \nonumber \\
& = \frac{4}{\pi^2} \frac{b^2-a^2}{a^2b^2} \frac{1}{s} \int_0^\infty \left[\iint\limits_0^\infty \E^{-s (u_1 + u_2)} \right.\nonumber \\ & \left.\times \cos\left(\frac{\xi}{b} u_1\right) \cos\left(\frac{\xi}{a} u_2\right) \D u_1 \D u_2\right] \frac{\xi \D\xi}{J^2_{\alpha}(\xi) + Y^2_{\alpha}(\xi)}, \nonumber
\end{align}
where we employed the Laplace representation of $s/(s^2 + \xi^2)$ given below \eqref{13/04/26-5}. Taking into account the formula $\cos(x_1)\cos(x_2) = \frac{1}{2}\big[\cos(x_1+x_2) + \cos(x_1-x_2)\big]$ and changing the order of integrals we get 
\begin{multline}\label{13/05/26-1}
g'_2(s)  = \frac{2}{\pi^2} \frac{b^2 - a^2}{a^2 b^2} \frac{1}{s} \iint\limits_0^{\infty} \E^{-s(u_1 + u_2)} \\ \times \left[T_1(u_1, u_2) + T_2(u_1, u_2)\right] \D u_1 \D u_2,
 \end{multline}
in which
 \begin{multline*}
T_1(u_1, u_2) = \int_0^{\infty} \cos\left[\xi \Big(\frac{u_1}{b} + \frac{u_2}{a}\Big)\right]\,\frac{\xi \D\xi}{J^2_{\alpha}(\xi)+Y^2_{\alpha}(\xi)} \quad \text{and} \\ T_2(u_1, u_2) = \int_0^{\infty} \cos\left[\xi \Big(\frac{u_1}{b} - \frac{u_2}{a}\Big)\right]\, \frac{\xi \D\xi}{J^2_{\alpha}(\xi)+Y^2_{\alpha}(\xi)}.   
 \end{multline*}
 The functions $T_1(u_1, u_2)$ and $T_2(u_1, u_2)$ are Fourier cosine transforms, so again, by the Tuck theorem \cite{Tuc06} they are positive spectral functions, because $\xi/\big(J^2_{\alpha}(\xi)+Y^2_{\alpha}(\xi)\big)$ is convex for all $\alpha\in[0,1]$ and $\xi>0$, see Lemma \ref{yy}.

Thereafter, we set $u_1 + u_2 = v$ in the curly brackets in the double integrals in Eq. \eqref{13/05/26-1}. That gives 
\begin{multline*}
\iint_0^\infty \E^{-s(u_1 + u_2)} T_j(u_1, u_2) \D u_1\D u_2 \\
= \int_0^\infty \E^{-s v} \left[\int_0^v T_j(v-u_2, u_2) \D u_2\right] \D v
\end{multline*}
for $j = 1, 2$. From Bernstein's theorem, it follows that it is a CM function since the square bracket is non-negative. Then, we can conclude that $g'(s)$ is a CM function as the product of CM functions. That finishes the proof. \qed

In the case of $\alpha = 1/2$, $$\widehat{f}_{2,2}(1/2; x, s) = \exp(-\lambda |x|/4)\, \exp\big[-b(\E^{\lambda |x|/2}-1) (\tau s^2+s)^{1/2} \big]$$ which is a CM function. That arises from Property 2, a Bernstein character of $(\tau s^2+s)^{1/2}$, and $b(\E^{\lambda\vert x \vert/2}-1)>0$.

\subsection{Normalization}\label{sect3.3}

\begin{theorem}
\label{22/05/26-1}
    The functions $p_{\lambda; 1}(\alpha; x, t)$ are $p_{\lambda; 2}(\alpha; x, t)$ are normalized with respect to $x\in\mathbb{R}$ .
\end{theorem}

\noindent
{\em Proof.} The normalization of $p_{\lambda; j}(\alpha; x, t)$, $j =1, 2$, can be shown by making the direct calculations where for the case of $j =1$ we set $y = \exp(-\lambda x/2)$ and use \cite[Eq. (7.14.1 (1)), p. 89]{AErdely53} whereas for $j = 2$ we set $y = \exp(\lambda x/2)$ and use \cite[Eq. (7.14.1 (3)), p. 89]{AErdely53}. \qed


\section*{Acknowledgments}
\v{Z}T was supported by NAWA under the ULAM programme. KG thanks the Director of the IFJ PAN Grant.


\end{document}